\documentclass[runningheads]{llncs}
\usepackage[T1]{fontenc}
\usepackage{graphicx, amsmath}
\usepackage{algorithm}
\usepackage{algpseudocode}
\usepackage{enumitem}
\usepackage{tikz}
\usepackage[appendix=inline]{apxproof}
\newtheoremrep{theorem}{Theorem}
\newtheoremrep{lemma}{Lemma}
\begin{document}
\title{Temporal Role Colouring}
%
%
\author{Jessica Enright \inst{1}\orcidID{0000-0002-0266-3292} \and
Kitty Meeks \inst{1}\orcidID{0000-0001-5299-3073} \and
Puck Rombach\inst{2}\orcidID{0000-0002-8374-0797} \and Ella Yates\inst{1} \orcidID{0009-0000-7979-0401}}
\authorrunning{J. Enright et al.}
%
\institute{University of Glasgow 
\email{\{jessica.enright,kitty.meeks\}@glasgow.ac.uk; e.yates.1@research.gla.ac.uk}
\and
University of Vermont
\email{puck.rombach@uvm.edu}\\
}
\maketitle              
\begin{abstract}
A \textit{role colouring} of a graph $G$ is an assignment of colours to the vertices of $G$ such that two vertices of the same 
colour have identical sets of colours in their neighbourhoods.  This model is used to capture the idea of vertices having roles in a contact network, consistent with the set of roles of their contacts.  
We define an extension of the role colouring problem to temporal graphs.
Temporal roles are defined via an automaton with states and transitions capturing both the current colour of a vertex and information about its current and past adjacencies. 
 
We show, by a reduction from the static problem, that the temporal role colouring problem is NP-complete. To contend with this intractability, we explore several parameterisations. We give fixed-parameter tractability results with respect to the number of states of the automaton combined with either the vertex-interval-membership width or the tree-interval-membership width of the temporal graph. We further show the problem is in FPT parameterised by the treewidth of the underlying graph, the lifetime of the temporal graph and the number of colours combined.

\end{abstract}
\section{Introduction}
Vertices representing different individuals in a network often have distinct roles within the network. For example a network of people at a restaurant might include customers, wait staff, managers and kitchen staff, as illustrated in Figure~\ref{fig:Roles example}. The role of an individual in this network determines which other individuals it may interact with: for example a customer will interact with at least one member of wait staff but not directly with the kitchen staff. 

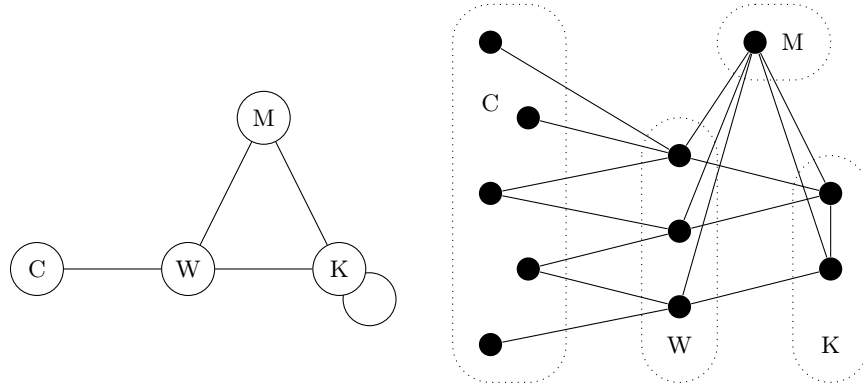
\begin{figure}
    \centering
    \begin{tikzpicture}[A/.style={circle,draw,fill=white,minimum size= 7mm},B/.style={circle,fill} ]
    \node[A](C) at (0,0){C};
    \node[A](Se) at (2,0) {W};
    \node[A](M) at (3,2) {M};
    \node[A](A) at (4.4,-0.4){};
    \node[A](K) at (4,0) {K};
    \draw (C)--(Se)--(M)--(K)--(Se);

    \node[B](M1) at (9.5,3){};
    \node[B](K1) at (10.5,0){};
    \node[B](K2) at (10.5,1){};
    \node[B](S1) at (8.5,-0.5){};
    \node[B](S2) at (8.5,0.5){};
    \node[B](S3) at (8.5,1.5){};
    \node[B](C1) at (6,-1){};
    \node[B](C2) at (6.5,0){};
    \node[B](C3) at (6,1){};
    \node[B](C4) at (6.5,2){};
    \node[B](C5) at (6,3){};
    \draw(C5)--(S3)--(M1)--(K1)--(K2)--(M1)--(S2)--(C3)--(S3)--(K2);
    \draw (C1)--(S1)--(C2);
    \draw(C2)--(S2)--(K2);
    \draw(K1)--(S1)--(M1);
    \draw (C4)--(S3);

    \draw[dotted,rounded corners=5mm] (5.5,3.5)--(5.5,-1.5)--(7,-1.5)--(7,3.5)--cycle;
    \node[] () at (6,2.2){C};
    \draw[dotted,rounded corners=5mm] (8,2)--(8,-1.5)--(9,-1.5)--(9,2)--cycle;
    \node[] () at (8.5,-1){W};
    \draw[dotted,rounded corners=5mm] (10,1.5)--(10,-1.5)--(11,-1.5)--(11,1.5)--cycle;
    \node[] () at (10.5,-1){K};
    \draw[dotted,rounded corners=5mm] (9,3.5)--(10.5,3.5)--(10.5,2.5)--(9,2.5)--cycle;
    \node[] () at (10,3) {M};

\end{tikzpicture}
    \caption{A role graph $R$ representing roles in a restaurant with vertices $\{C,K,M,W\}$ representing customers, kitchen staff, managers and wait staff respectively and a graph that is $R$-role colourable.}
    \label{fig:Roles example}
\end{figure}

This idea is formalised by the notion of role colouring.
A \textit{role colouring} of a graph is an assignment of colours (which correspond to roles) to the vertices such that two vertices are the same colour if they have exactly the same set of colours in their neighbourhoods. 
The study of roles originates  from sociology where they were used to examine specific functions and interactions within social groups \cite{borgatti2024analyzing,lorrain1971structural}. They have also been used to study online social networks \cite{golder2004social}, biological networks \cite{luczkovich2003defining} and technological networks \cite{scripps2007node}.
White and Reitz \cite{white1983graph} formalised the structural approach to the problem in terms of graph homomorphisms.
The concept of a \textit{role  graph} has also been studied, that is a graph $R$ that defines what an acceptable role colouring looks like. A graph $G$ is \textit{$R$-role colourable} if there is a colouring $c$ of the vertices such that $c(N_G(v))=N_R(c(v))$ for every vertex $v \in V(G)$, where $N_G$ and $N_R$ denote the neighbourhoods in $G$ and $R$ respectively. 
The $R$-role colouring problem is defined as follows.\\

\noindent
\textsc{$R$-role colouring}\\
\textbf{Input:} Graphs $G=(V_G,E_G)$ and $R=(V_R,E_R)$\\
\textbf{Question:} Is there a function $r:V_G \rightarrow V_R$ such that for all $u\in V_G: r(N_G(u))=N_R(r(u))$.\\

The \textsc{$R$-role colouring} problem is NP-complete for any simple connected graph $R$ on at least three vertices \cite{fiala2005complete}. A related problem, $k$-role colouring, asks if there is \emph{any} simple connected graph $R$ on $k$ vertices such that a graph $G$ is $R$-role colourable.   $k$-role colouring is NP-hard even when $G$ is a planar graph but can be solved in polynomial time when restricted to trees \cite{purcell2015complexity}.

The existing literature assumes that roles and hence types of interaction are fixed over time but this may not be realistic: it is known that modern online social networks are highly dynamic \cite{kumar2006structure}. In our earlier toy example a customer may apply for a job on the wait staff and eventually be promoted to manager meaning their interactions also evolve over time. To model this we introduce the notion of temporal role colouring. Here the interactions are given as a temporal graph in which edges are subject to discrete changes over time and instead of a static role graph we use a role automaton to capture the permitted interactions and transitions between roles.

In the remainder of this section we introduce the necessary concepts and notation relating to both temporal graphs and automata to allow us to give a formal definition of our temporal problem, before summarising our main results.

\paragraph*{Temporal graphs and parameters}
A temporal graph $\mathcal{G}$ consists of a static graph, sometimes called the \emph{underlying graph} $G_\downarrow=(V,E)$ and an assignment of times $\lambda:V \rightarrow 2^\Lambda$ where $\Lambda$ is the latest time at which an edge is active, called the \emph{lifetime}. We refer to the set $\mathcal{E}$ given by all pairs $(e,t)$ with $t \in \lambda(e)$ as the \emph{time edges}.
 The static graph at time $t$ with vertices $V$ and all time edges active at time $t$ is called the \emph{snapshot} of $\mathcal{G}$ at time $t$, also written as $G_t$. We use the notation $N^t(v)$ to refer to the neighbourhood of a vertex $v$ at time $t$.

 Even problems that can be solved in polynomial time on static graphs are often intractable on temporal graphs. This has led to the study of restrictive temporal parameters, including vertex-interval-membership width \cite{bumpus2023edge}.

 \begin{definition}
    %
    The \emph{vertex-interval-membership sequence} of the temporal graph $\mathcal{G}$ is the sequence $(U_t)_{t\in[\Lambda]}$ of vertex-subsets of $\mathcal{G}$ where $U_t:=\{v\in V(G):\exists i,j \in [\Lambda] \text{ and } uv,wv \in E(G) \text{ such that }
    i\in \lambda(uv), j\in\lambda(wv) \text{ and } i\leq t\leq j\}$. Then the \emph{vertex-interval-membership width} is $\omega(\mathcal{G}):=\max_{t\in \Lambda}|U_t|$.  
\end{definition}

We also use an extension of this, the \textit{tree-interval-membership width}, formally defined in Section $4$.
 
\paragraph*{Automata}
In general, we use notation and algorithmic results on automata from \cite{esparza2023automata}.   An automaton consists of states and transitions between them, and gives traces through the states according to an input string. 

In our automata there are some cases where a state may change spontaneously with no input read; in these cases the transition is called an \textit{$\varepsilon$-transition} and labelled with the empty word $\varepsilon$.
More formally this kind of automaton can be defined as follows.
\begin{definition}[\cite{esparza2023automata}]
    A \emph{nondeterministic automaton} with $\varepsilon$-transitions is a tuple $A=(Q,\Sigma,\delta,Q_o,F)$, where
\begin{itemize}
\item $Q$ is a nonempty set of states
\item $\Sigma$ is the input alphabet
   \item $\delta:Q \times (\Sigma \cup \{\varepsilon\}) \rightarrow \mathcal{P}(Q)$
\item $q_0 \in Q$ is the initial state
\item $F \subseteq Q$ is the set of final states.
\end{itemize}
\end{definition}
We sometimes refer to the transitions as triples $(q_i,a,q_j)$ to mean that $q_j \in \delta(q_i,a)$ where $\delta$ is the transition function of our automaton. 
Following Algorithm $13$ in \cite{esparza2023automata} we assume here that checking if a transition $(q_i,a,q_j)$ exists can be done in time $\mathcal{O}(|Q|)$. We refer to the number of states $|Q|$ as $\rho$. We also use the assumptions given in \cite{esparza2023automata} that, for a given state $q$, checking if $q=q_0$ and if $q\in F$ can be done in constant time.

\paragraph{Formal problem definition} While static role colouring uses a graph to describe the required relationships between colours, for temporal role colouring we use an automaton to encode the colours (recall these correspond to roles) and transitions between them. States of the automaton and colours are not synonymous: each state has a colour, but it may also encode other information such as what colours vertices in the state have already seen.
We use a nondeterministic automaton to allow modelling of situations such as a vertex being able to change colours at any time after seeing specific neighbourhood colours, rather than requiring change immediately.  
In our earlier restaurant example we may think of a customer joining the wait staff only after having been interviewed by a manager: then we may have a state for customers and another state for customers who have been interviewed but both would be considered to have the same customer role (or colour in the automaton).

\begin{definition}[$k$-colour role automaton]
For an integer $k$,
$A$ is a \emph{$k$-colour role automaton} if it is a finite automaton with an additional colouring of states $c$. Specifically it is a tuple containing:
\begin{itemize}
    \item An input alphabet $\Sigma =\mathcal{P}([k])$
    \item A set $Q$ of states
    \item A colouring of states, $c:Q\to \{0,..,k\}$
    \item A set $F \subseteq Q$ of accepting states
    \item A start state $q_0$ such that $c(q_0)=0$
    \item A transition function $\delta: Q \times \Sigma \rightarrow \mathcal{P}(Q)$ where all transitions from $q_0$ are $\varepsilon$-transitions and there are no $\varepsilon$-transitions from any other state.
\end{itemize}
\end{definition}
Given a temporal graph $\mathcal{G}=(G,\lambda)$ with lifetime $\Lambda$ and a $k$-colour role automaton $A$ that includes colouring $c$ we say that $\mathcal{G}$ \textit{admits a temporal role colouring} by $A$ if there exists a mapping $r:(V(G)\times [\Lambda+1])\rightarrow [k]$ such that for each vertex $v$ the word $ r(N^1(v))r(N^2(v))...r(N^\Lambda(v))$ is an accepting word for $A$ with states $q_1,...q_{\Lambda+1}$ such that $r(v,t)=c(q_t)$. Here $N^t(v)$ is the set of neighbours of $v$ at time $t$, so the word is a sequence of sets of colours. 
This leads us to a temporal equivalent of the static problem.\\

\noindent
\textsc{temporal role colouring}\\
\textbf{Input:} Temporal graph $\mathcal{G}=(G,\lambda)$, a $k$-colour role automaton $A$.\\
\textbf{Question:} Does $\mathcal{G}$ admit a temporal role colouring by $A$?\\

\paragraph{Our contribution}
We show that \textsc{temporal role colouring} is NP-complete in general but also derive fpt-algorithms for several restricted cases. 
In Section $3$ we give an fpt-algorithm for the decision problem parameterised by the vertex-interval-membership width and the number of states in the automaton, via dynamic programming over the vertex-interval-membership sequence.
In Section $4$ we show that our problem has the necessary properties to apply a meta-algorithm to solve the problem in FPT time when parameterised by the tree-interval-membership width and the number of states in the automaton.
Finally, in Section $5$ we give an fpt-algorithm parameterised by the underlying treewidth of the graph, the number of colours $k$ used by the automaton and the lifetime $\Lambda$ of the temporal graph; this is achieved by dynamic programming over a nice tree decomposition.

\section{NP-Completeness}
 In this section we give a reduction from the static \textsc{R-role colouring} problem to show that the \textsc{temporal role colouring} problem is also NP-complete.

 
\begin{theoremrep}
\textsc{temporal role colouring} is NP-complete.   
\end{theoremrep}
\begin{proof}
First we show that the problem belongs in NP. Given a colouring of $\mathcal{G}$, the word consisting of the colours of the neighbours at each time for a given vertex can be generated in time $\mathcal{O}(\Lambda n)$. 
Algorithm $13$ of \cite{esparza2023automata} tells us that an accepting word can be checked in time $\mathcal{O}(\Lambda \rho^2)$. Combining these gives us that a yes instance can be verified in time $\mathcal{O}(n^2\Lambda^2\rho^2)$, this is polynomial in the input, therefore \textsc{temporal role colouring} is in NP.

We show NP-hardness by a reduction from the static problem \textsc{$R$-role colouring}. Consider an instance $(G,R)$ of \textsc{$R$-role colouring} such that $G=(V_G,E_G)$ and $R=(V_R,E_R)$ where $|V_R|=k$; we can assume the vertex set $V_R=[k]$ without loss of generality. We construct a role automaton $A_R$ from the role graph $R$ by setting:
\begin{itemize}
    \item Input alphabet $\Sigma = 2^{[k]}$
    \item A start state $q_0$
    \item A set of states $Q=\{q_0,q_{1,s},q_{1,a},q_{2,s},q_{2,a},...,q_{k,s},q_{k,a}\}$,
 that is, the start state and two states for each vertex in $R$
    \item A colouring of states $c:Q \rightarrow \{0,..,k\}$ defined as $c(q_0)=0$ and $c(q_{i,a})=c(q_{i,s})=i$
    \item A set $F=\{q_{1,a},q_{2,a},...q_{k,a}\}$ of accepting states
    \item A transition function $\delta$ with $\delta(q_0,\varepsilon)=q_{i,s}$ for each $i \in [k]$ and $\delta(q_{i,s},N_R(i))=q_{i,a}$ for each $i \in [k]$.
\end{itemize}
We define the temporal graph $\mathcal{G}$ as the temporal graph with lifetime $1$ such that the snapshot $G_1=G$.
This is clearly a polynomial-time construction of a temporal role colouring instance.
We claim that $(\mathcal{G},A_R)$ is a yes instance of \textsc{temporal role colouring} if and only if $(G,R)$ is a yes instance of \textsc{$R$-role colouring}.

Let $(\mathcal{G},A_R)$ be a yes instance of \textsc{temporal role colouring} with the colouring $r:V_G \times[2] \rightarrow k$.
We consider an arbitrary vertex $v \in V_G$. We let $r_1$ be function such that $r_1(v)=r(v,1)$.
As $(\mathcal{G},A_R)$ is a yes instance we know the word $r(N^1(v))=r_1(N_G(v))$ is an accepting path in the automaton $A_R$ with states $q_x$ and $q_y$ such that $r(v,1)=c(q_x)$ and $r(v,2)=c(q_y)$. This means we must have a transition $\delta(q_x,r_1(N_G(v)))=q_y$. By our construction of the transition function this gives us that for some $i$ we must have that $q_x=q_{i,s}$, $q_y=q_{i,a}$ and $r_1(N_G(V))=N_R(i)$. We know from the construction that $c(q_{i,s})=i$, therefore we have that $r_1(v)=i$. As $v$ was arbitrary we have for all $v \in V_G$ that $r_1(N_G(v))=N_R(r_1(v))$ which is exactly the requirement for $(G,R)$ to be a yes instance of \textsc{$R$-role colouring}.

Now let $(G,R)$ be a yes instance of \textsc{$R$-role colouring}, therefore there is a function $r: V_G \rightarrow V_R$ such that for all $v \in V_G$ we have $ r(N_G(v))=N_R(r(u))$. 
We define the temporal colouring $s: V_G \times [2] \rightarrow [k]$ as $s(v,t)=r(v)$ for all $v \in V_G$, $t \in [2]$.
Now fix $v$ as an arbitrary vertex of $V_G$.
As $V_R=[k]$ there is some $i \in [k]$ such that $r(v)=i$ and hence $r(N_G(v))=N_R(i)$. 
By construction of the automaton an accepting word for $v$ is given by $N_R(i)$, which corresponds to the path through the states $q_{i,s}$, $q_{i,a}$ where $c(q_{i,s})=c(q_{1,a})=i$.
 Therefore, for all $v\in V_G$ we have that $r(N_G(v))=s(N^1(v))$ is an accepting word in the automaton with states $q_{i,s}$ and $q_{i,a}$ such that $s(v,1)=c(q_{i,s})$ and $s(v,2)=c(q_{i,a})$, which is exactly the requirement for $(\mathcal{G},A_R)$ to be a yes instance of \textsc{temporal role colouring}.
Hence, the \textsc{temporal role colouring} problem is NP-complete.\qed
\end{proof} 

\begin{toappendix}

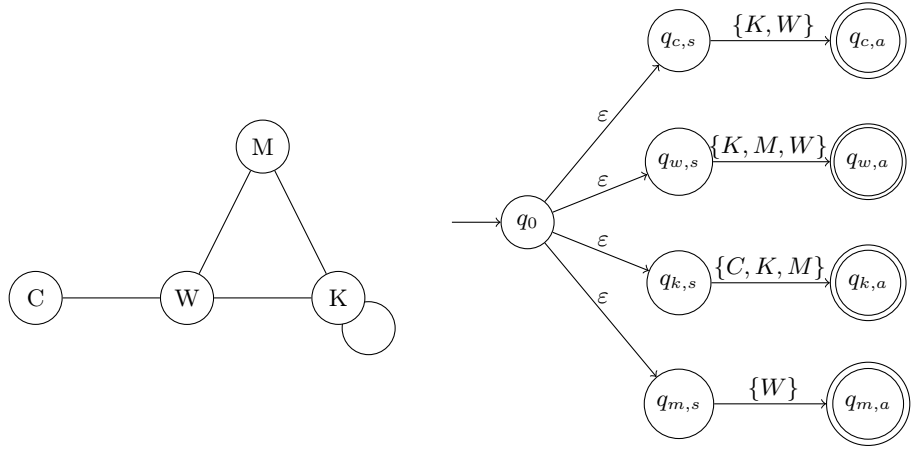
\begin{figure}[H]
    \centering
    \begin{tikzpicture}[A/.style={circle,draw,fill=white,minimum size= 7mm},B/.style={circle,fill}, C/.style={circle,draw,minimum size=10mm}, D/.style={circle,draw,minimum size=11mm} ]
    \node[A](C) at (0,0){C};
    \node[A](Se) at (2,0) {W};
    \node[A](M) at (3,2) {M};
    \node[A](A) at (4.4,-0.4){};
    \node[A](K) at (4,0) {K};
    \draw (C)--(Se)--(M)--(K)--(Se);

    \node[A](q0) at (6.5,1){$q_0$};
    \node[A](C1) at (8.5,3.4){$q_{c,s}$};
    \node[A]() at (11,3.4){$q_{c,a}$};
    \node[C](C2) at (11,3.4){};
    \node[A](M1) at (8.5,1.8){$q_{w,s}$};
    \node[A]() at (11,1.8){$q_{w,a}$};
    \node[C](M2) at (11,1.8){};
    \node[A](S1) at (8.5,0.2){$q_{k,s}$};
    \node[A]() at (11,0.2){$q_{k,a}$};
    \node[C](S2) at (11,0.2){};
    \node[A](K1) at (8.5,-1.4){$q_{m,s}$};
    \node[A]() at (11,-1.4){$q_{m,a}$};
    \node[D](K2) at (11,-1.4){};
    \draw[->](5.5,1)->(q0);
    \draw[->](q0)->(C1);
    \draw[->](q0)->(M1);
    \draw[->](q0)->(S1);
    \draw[->](q0)->(K1);
    \draw[->](C1)->(C2);
    \draw[->](M1)->(M2);
    \draw[->](K1)->(K2);
    \draw[->](S1)->(S2);
    \node[]() at (7.5,2.4){$\varepsilon$};
    \node[]() at (7.5,1.55){$\varepsilon$};
    \node[]() at (7.5,0.75){$\varepsilon$};
    \node[]() at (7.5,0){$\varepsilon$};
    \node[] () at (9.75,-1.2){$\{W\}$};
    \node[] () at (9.7,0.4){$\{C,K,M\}$};
    \node[] () at (9.68,2){$\{K,M,W\}$};
    \node[] () at (9.75,3.6){$\{K,W\}$};
    
\end{tikzpicture}
\caption{The role graph $R$ and the automaton $A_R$ constructed by the hardness reduction.}
\end{figure}
\end{toappendix}
\section{Vertex-Interval-Membership Width}
In this section we give an fpt-algorithm for \textsc{temporal role colouring} with respect to the vertex-interval-membership width $\omega$ of the temporal graph $\mathcal{G}$ and the number of states $\rho$ of the automaton $A$. Bumpus and Meeks \cite{bumpus2023edge} give an algorithm which computes $\omega$ in time $\mathcal{O}(\omega\Lambda)$ where $\Lambda$ is the latest time an edge in $\mathcal{G}$ is active.

For each time $i$ with $0 \le i \le \Lambda$, let $Q_i$ denote the set of states in which a vertex may be when it first becomes active at time $i$.
Similarly we use $F_i$ to refer to the set of states from which an inactive vertex may reach an accept state in $i$ time steps (i.e. for all $q \in F_i$ the word $\emptyset^{i}$ reaches an accept state).

\begin{lemmarep}\label{lem:vim sets}
    For a given instance $(\mathcal{G},A)$ the sets $Q_0,...,Q_\Lambda$ and $F_0,...,F_\Lambda$ can be found in time $\mathcal{O}(\rho^2(\Lambda +\rho))$ where $\Lambda$ is the lifetime of the temporal graph and $\rho$ is the number of states in the automaton.
\end{lemmarep}
\begin{proof}
We create these sets of states from the states $Q$ and the transition function $\delta$ of the automaton. We do this by first creating a directed graph with vertices $V=Q$ and the directed edge set $E$ given by $q_i\rightarrow q_j\in E$ if $q_j \in \delta(q_i,\emptyset)$. This can be constructed in time $\mathcal{O}(\rho^3)$ as we check the existence of a transition for $\rho^2$ pairs of states.

We let $Q_0=\{q \in Q:q \in \delta(q_0,\varepsilon)\}$,
for all $i \in [\Lambda]$ we construct $Q_i=\{q:p\rightarrow q \in E \text{ for some } p \in Q_{i-1}\}$. For each $i$ this is clearly the set of states a vertex may reach at time $i$ with the word $\emptyset^i$, this is exactly the set of colours in the neighbourhood for a vertex that has not yet had an edges incident to it, making $Q_i$ the set we claim it to be.

Similarly, we let $F_0=F$ (the set of accepting states from our automaton) and for all $i\in [\Lambda]$ we define $F_i=\{q: q \rightarrow p \in E \text{ for some } p \in F_{i-1}\}$. For each $i$ this is clearly the set of states from which a vertex can reach an accepting state in time $i$ with the word $\emptyset^i$. The is exactly the word composed of sets of colours in the neighbourhood of a vertex that will not have any edges incident to it. $F_i$ is exactly the set we claim it to be.

We can construct all of these sets in time $\mathcal{O}(\Lambda \rho^2)$. Therefore, the total time to construct the directed graph and find the sets $Q_i$s and $F_i$s is $\mathcal{O}(\Lambda\rho^2+\rho^3)$.
\qed \end{proof}

\begin{toappendix}
Let $\mathcal{G}=(V,\mathcal{E})$ and let $A$ be a $k$-colour automaton with states $Q$ and transition function $\delta$, then for the instance $(\mathcal{G},A)$
we consider an assignment at time $t$, $\alpha_t: V \rightarrow Q$, \textbf{valid} on a set of vertices $X \subseteq V$ if there is a colouring $c: V \times [t] \rightarrow [k]$ such that:
\begin{enumerate}
    \item for all vertices $v$ in $X$ the word $c(N^1(v))...c(N^{t-1}(v))$ has a path through $A$ with states $q_1,...q_{t}$ such that $\alpha(v)= q_{t}$
    \item for all vertices $v \in (U_1 \cup U_2 \cup ... \cup U_{t-1}) \setminus U_t$ the word $c(N^1(v))...c(N^{t-1}(v))\emptyset^{\Lambda-t+1}$ has an accepting path through $A$.
    \item if $v \in X \cap ((U_1 \cup U_2 \cup ... \cup U_{t-1})\setminus U_t)$ the accepting path has states $q_1,...,q_t,...,q_\Lambda$ such that $\alpha_t(v)=q_t$
\end{enumerate} 

We say an assignment at time $t+1$ is \textbf{supported} by an assignment at time $t$ if for all $v\in U_t$, $\alpha_{t+1}(v) \in \delta(\alpha_t(v), c(N^t(v))$. 

We now give an algorithm to determine if a state at time $t+1$ is valid and supported by a given valid state at time $t$. To check that it is supported we check that there is a transition between states labelled by the colours of the neighbourhood for every vertex in $U_t$. To check validity we check that all vertices joining the set $(U_1 \cup... \cup U_t)\setminus U_{t+1}$ have accepting words in the automaton.

 \begin{algorithm}
 \caption{}
        \label{alg:vim}
        \begin{algorithmic}[1]
        \Require $\mathcal{G}=(G,\lambda)$ is a temporal graph on $n$ vertices with the vertex-interval-membership sequence $U_{1},...,U_\Lambda$; $A$ is a $k$-colour role automaton with colouring $c$ and transition function $\delta$; $\alpha_{t}$ is a valid assignment of the vertices $U_{t}$ at time $t$; $\alpha_{t+1}$ is an assignment of the vertices $U_{t}$ at time $t+1$; and $F_{\Lambda-t+1}$ is the set of states which accept the word $\emptyset^{\Lambda-t+1}$.
    
    \ForAll{ $v\in U_{t}$}
    \If{$\alpha_{t+1}(v) \notin \delta(\alpha_{t}(v), c(N^t(v)))$}
    \State \Return \textbf{false}
    \EndIf
    \EndFor
    \ForAll{$v \in U_{t} \setminus U_{t+1}$}
    \If{$\alpha_{t+1}(v) \notin F_{\Lambda-t}$}
    \State \Return \textbf{false}
    \EndIf
    \EndFor
    \State \Return \textbf{true}
        \end{algorithmic}
        \end{algorithm}
        
\begin{lemma} \label{lem:algStep}
Algorithm \ref{alg:vim}$(\mathcal{G},A,\alpha_t,\alpha_{t+1},F_{\Lambda-t+1})$ returns true if and only if $\alpha_{t+1}$ is a valid assignment on $U_t$ supported by $\alpha_t$.
\end{lemma}
\begin{proof}
First, we let Algorithm $1$ return true.
Line $3$ would return false if for any vertex $v \in U_t$ the automaton does not have a transition $(\alpha_t(v),c(N^{t}(v)),\alpha_{t+1}(v))$, therefore for all vertices $v \in U_t$ we have that $\alpha_{t+1}(v) \in \delta(\alpha_t,c(N^t(v)))$ and $\alpha_{t+1}$ is supported by $\alpha_t$.
Validity of $\alpha_t$ on $U_t$ tells us that for all $v \in U_t$ there is a path through the automaton labelled $c(N^1(v))...c(N^{t-1}(v))$ through states $q_1,...,q_{t}$ such that $\alpha_{t}(v)=q_{t}$. Combining validity of $\alpha_t$ and the observation that $\alpha_{t+1}$ is supported by $\alpha_t$ gives us that the word $c(N^1(v))...c(N^{t-1}(v))c(N^{t}(v))$ has a path through the automaton with states $q_1,...,q_{t+1}$ such that $\alpha_{t+1}(v)=q_{t+1}$ and Condition $1$ holds.

Similarly, if Algorithm \ref{alg:vim} returns true then line $8$ did not return false and for all $v\in U_t \setminus U_{t+1}$ we have that $\alpha_t(v) \in F_{\Lambda-t+1}$. The set $F_{\Lambda-t}$ is defined as the set of states in $Q$ from which the word $\emptyset^{\Lambda-t+1}$ has a path to an accepting state. As $v \in U_t$ we know from the above argument that the word $c(N^1(v))...c(N^{t}(v))$ has a path through the automaton with states $q_1,...,q_{t+1}$ such that $\alpha_{t+1}(v)=q_{t+1} \in F_{\Lambda-t}$. 
Which gives us that for all $v \in U_t \setminus U_{t+1}$ the word $c(N^1(v))...c(N^{t}(v))\emptyset^{\Lambda-t}$ is accepted by the automaton  as required and that Condition $3$ holds.
Now we consider the set $(U_1 \cup ... \cup U_{t-1}) \setminus U_t$.
From validity of $\alpha_t$ we know that for all vertices $v \in (U_1 \cup ... \cup U_{t-1}) \setminus U_t$ the word $c(N^1(v))...c(N^{t-1}(v))\emptyset^{\Lambda-t+1}$ is accepted. As $v \notin U_t$ we know that $N^t(v)=\emptyset$ and therefore $c(N^t(v))=\emptyset$ which gives us that $c(N^1(v))...c(N^{t-1}(v))\emptyset^{\Lambda-t+1}=c(N^1(v))...c(N^t(v))\emptyset^{\Lambda-t}$ and therefore that the word $c(N^1(v))...c(N^t(v))\emptyset^{\Lambda-t}$ is accepted by the automaton.
We note that $(U_1 \cup ... \cup U_{t-1} \cup U_t)\setminus U_{t+1} = ((U_1 \cup ... \cup U_{t-1}) \setminus U_t) \cup (U_t \setminus U{t+1})$ as from the definition of the vertex-interval-membership sequence no vertex may be in $U_i$ with $i<t$ and $U_{t+1}$ without also being in $U_t$. Therefore, for all vertices $v \in (U_1 \cup ... \cup U_{t-1} \cup U_t)\setminus U_{t+1}$ the word $c(N^1(v))...c(N^t(v))\emptyset^{\Lambda-t}$ is accepted by the automaton, that is that Condition $2$ holds.
Hence, if Algorithm $1$ returns true, the assignment $\alpha_{t+1}$ is valid and supported by $\alpha_t$.

Now we let $\alpha_{t+1}$ be a valid assignment on $U_t$ supported by $\alpha_t$. By the definition of supported the if statement on $2$ is never satisfied and line $3$ does not return false. 
Next, we note that the set given in Condition $3$ is $U_t \cap ((U_1 \cup... \cup U_t)\setminus U_{t+1})=U_t \setminus U_{t+1}$. Hence, for all $v \in U_t \setminus U_{t+1}$ the word $c(N^1(v))...c(N^{t}(v))\emptyset^{\Lambda-t}$ has an accepting path through the automaton with states $q_1,...,q_{t+1},...,q_\Lambda$ such that $\alpha_{t+1}(v)=q_{t+1}$. Therefore, for all $v \in U_t \setminus U_{t+1}$ there is accepting path from the state $\alpha_{t+1}(v)$ labelled by $\emptyset^{\Lambda-t}$, that is $\alpha_{t+1}(v) \in F_{\Lambda-t}$, which means the condition on line $7$ is never satisfied and line $8$ never returns false. Hence, when $\alpha_{t+1}$ is valid and supported by $\alpha_t$ Algorithm $1$ returns true.
\qed \end{proof}

 \begin{algorithm}
 \caption{VIM- width algorithm}
        \label{alg:vim2}
        \begin{algorithmic}[1]
        \Require $\mathcal{G}=(G,\lambda)$ is a temporal graph on $n$ vertices with the vertex-interval-membership sequence $U_{1},...,U_\Lambda$;
        $A$ is a $k$-colour role automaton with states $Q$, colouring $c$ and transition function $\delta$; $Q_0,...,Q_\Lambda$ are possible starting states at each time; and $(F_t)_{0\leq i \leq \Lambda}$ is the sequence of sets of states with the word $\emptyset^i$  accepting.
        
    \State $\mathcal{A}^+_1 \gets\{\alpha_1:U_1 \rightarrow Q_0\}$
    \Comment{All functions assigning vertices in $U_1$ to states in $Q_0$}
    \ForAll{ $i \in [\Lambda-1]$}
    \State $\mathcal{A}_{i+1} \gets \emptyset$
    \State $\Gamma_{i+1} \gets \{\alpha_{i+1}:U_{i} \rightarrow Q\}$
    \ForAll{$\alpha_{i+1} \in \Gamma_{i+1}$}
    \ForAll{$\alpha_{i} \in \mathcal{A}_{i}^+$}
    \If{Algorithm \ref{alg:vim}($\mathcal{G},A,\alpha_{i},\alpha_{i+1},F_{\Lambda-t+1}$)=\textbf{true}}
    \Comment{Check if state is supporting}
    \State Set $\mathcal{A}_{i+1}= \mathcal{A}_{i+1} \cup \{\alpha_{i+1}\}$
    \State Break 
    \Comment{Advance to next $\alpha_{i+1} \in \Gamma_{i+1}$}
    \EndIf
    \EndFor
    \EndFor
    \If {$\mathcal{A}_{i+1}=\emptyset$}
    \State \Return \textbf{false}
    \EndIf
    \State $\mathcal{A}_{i+1}^+\gets  \mathcal{A}_{i+1} \times \{\beta:U_{i+1}\setminus U_{i} \rightarrow Q_i\}$
    \Comment{All combinations with assignments of newly active vertices}
    \EndFor
\ForAll{$\alpha_\Lambda \in \mathcal{A}^+_\Lambda$}
\ForAll {$ v \in U_\Lambda$}
\If {$\delta(\alpha_\Lambda(v),c(N^\Lambda(v))) \cap F_0 =\emptyset$}
\State Break
\Comment{Advance to next $a_\Lambda \in \mathcal{A}^+_\Lambda$}
\EndIf
\EndFor
\State \Return \textbf{true}
\EndFor
\State \Return \textbf{false}

        \end{algorithmic}
        \end{algorithm}

\end{toappendix}

Using Lemma~\ref{lem:vim sets} we derive a dynamic program over the vertex-interval-membership sequence to solve temporal role colouring.
This requires us to show we can assess validity of a possible role colouring assignment in terms of the possibilities at the previous time step efficiently. 

\begin{theoremrep} \label{thm:vim}
    For a temporal graph $\mathcal{G}$ and a $k$-colour role automaton $A$, \textsc{temporal role colouring} on $\mathcal{G}$ is solvable in time $\mathcal{O}(\Lambda\rho^{3\omega+1}\omega^2)$ where $\omega$ is the vertex-interval-membership width of $\mathcal{G}$ and $\rho$ is the number of states of $A$.
\end{theoremrep}
\begin{proof}

We argue that Algorithm \ref{alg:vim2} determines if $(\mathcal{G},A)$ is a yes instance.
We note that line $1$ gives all possible assignments of states to the set $U_1$ at time $0$. 

For all times $t \geq 1$, the for loops starting on lines $5$ and $6$ check for each possible assignment of the vertices in $U_t$ if there is a valid assignment of $U_t$ at time $t$ such that Algorithm \ref{alg:vim} returns true.
We know from Lemma \ref{lem:algStep} that Algorithm \ref{alg:vim} returns true if and only if $\alpha_{t+1}$ is a valid assignment on $U_t$ supported by $\alpha_t$, it is then added to the list of valid assignments for time $t+1$. 
All valid assignments on $U_t$ at time $t+1$ must be supported by some valid assignment on $U-t$ set at time $t$: 
from Condition $1$ of validity for each $v\in U_t$ the word $c(N^1(v))...c(N^{t}(v))$ has a path through $A$ with states $q_1,...,q_{t+1}$
it can clearly be seen that for all the conditions to hold they must also hold for the set of assignments $\alpha_1,...,\alpha_t$ given by $\alpha_i(v)=q_i$ for each $v \in U_t$.

Therefore, if after iterating through all valid assignments at time $t$, there is no assignment that supports $\alpha_{t+1}$ we know it is not a valid assignment.
Hence, after iterating through $\Gamma_{t+1}$ we know the set $\mathcal{A}_{t+1}$ is exactly the set of valid functions at time $t+1$ on $U_{t}$. 

The set of vertices $U_{t+1}\setminus U_{t}$ is exactly the vertices starting their active period at time $t+1$. The initial states these vertices may be in is given by $Q_t$. Therefore, we can see that $\mathcal{A}^+_{t+1}$ is exactly the set of valid functions at time $t+1$ on $U_{t+1}$.

At time $\Lambda$ we then have all valid assignments of $U_\Lambda$. For each assignment we check on line $20$ if every vertex has a transition labelled by its neighbourhood to an accepting state. If all vertices in $U_t$ do then we know from validity of the assignment that for every non-isolated vertex the word $c(N^1(v))...c(N^\Lambda(v))$ is accepted by the automaton. For completely isolated vertices we know the word  $c(N^1(v))...c(N^\Lambda(v))$ is accepting if $Q_0 \cap F_\Lambda \neq \emptyset$. If the algorithm returns true, and any isolated vertex is accepted, then the colouring given by the colouring of the states $c$, is the colouring required for $(\mathcal{G},A)$ to be a yes instance.
If not there is no assignment where every vertex reaches an accept state, therefore there is no possible colouring of $\mathcal{G}$ such that all vertices have the required accepting words and $(\mathcal{G},A)$ is a no instance.
Hence, Algorithm \ref{alg:vim2} determines \textsc{temporal role colouring}.

We now consider run time. 
By Lemma \ref{lem:vim sets} constructing the sets $Q_0,...,Q_\Lambda$ and $F_0,...,F_\Lambda$ takes time $\mathcal{O}(\rho^2(\Lambda+\rho))$ where $\rho$ is the number of states of the automaton.

The size of the set $\mathcal{A}^+_t$ is at most $\rho^{2\omega}$, where $\omega$ is the vertex-interval-membership width, as it assigns vertices in $U_t$ and $U_{t+1}$ to states in $Q$. The size of the set $\Gamma_{t+1}$ is at most $\rho^\omega$.
At each time $t \in [\Lambda-1]$, we run Algorithm \ref{alg:vim} at most $\rho^{3\omega}$ times. 
The first loop of Algorithm $1$ iterates over vertices in $U_t$, finds the colours of each neighbourhood and checks if a state is in the transition given, it therefore runs in time $\mathcal{O}(\rho\omega^2)$ as each of the at most $\omega$ vertices has at most $\omega$ neighbours.
The second loop runs in time $\mathcal{O}(\rho\omega)$. Hence, Algorithm \ref{alg:vim} runs in time $\mathcal{O}(\rho\omega^2)$.
This means in Algorithm \ref{alg:vim2} the for loop on lines $2$ to $16$ runs in time $\mathcal{O}(\Lambda\rho^{3\omega+1}\omega^2)$.
The next loop iterates over at most $\rho^{2\omega}$ assignments and for each vertex in $U_\Lambda$ does a comparison of sets each with size at most $\rho$. Therefore this for loop runs in time $\mathcal{O}(\rho^{2\omega+1}\omega)$.

Finally, if there are any isolated vertices we can check they are accepted in time $\mathcal{O}(\rho)$.

This gives us that the total running time is $\mathcal{O}(\rho^2(\Lambda+\rho)+\Lambda\rho^{3\omega+1}\omega^2+\rho^{2\omega+1}\omega+\rho)= \mathcal{O}(\Lambda\rho^{3\omega+1}\omega^2)$.
\qed \end{proof} 

\section {Tree-Interval-Membership Width}
In this section we use definitions and results from Enright et al. \cite{enright2025families} to prove that \textsc{temporal role colouring} is in FPT with respect to the number of states in the automaton and the tree-interval-membership width of the temporal graph $\mathcal{G}$. 
TIM width is a generalisation of vertex-interval-membership width that takes into account connectivity. A TIM decomposition allows us to work on disconnected components of a snapshot simultaneously; this suits our problem as each transition is dependent only on neighbours of vertices, so disconnected components have no effect on each other. 

\begin{definition}[Tree-Interval-Membership Width]
We say a triple $(T,B,\tau)$ is a tree-interval-
membership decomposition (TIM decomposition) of a temporal graph $\mathcal{G}$ with lifetime $\Lambda$ if $T$ is a labelled directed tree, where $B={B(s):s \in V(T)}$ is a collection of subsets of $V(\mathcal{G})$, called bags, and $\tau :V(T) \rightarrow [\Lambda]$ is a function which labels each node with a time, satisfying:
\begin{enumerate}
    \item For all vertices $v\in V(\mathcal{G})$ and times $t\in [\Lambda]$, there exists a unique node $i\in V(T)$ such that $\tau (i)=t$ and $v \in B(i)$.
    \item For all time-edges $(uv,t) \in \mathcal{E}(G)$, there exists a node $i \in V(T)$ such that $\{u,v\}\subseteq B(i)$ and $\tau(i)=t$. 
    \item The set of arcs of $T$ is given by $\{(i,j) : B(i)\cap B(j) \neq  \emptyset \text{ and } \tau (j)= \tau (i)+1\}$.
\end{enumerate}
The \emph{width} of a TIM decomposition is defined to be $ \max \{ |B(s)| :s \in V(T)\}$. The \emph{TIM width} of a temporal graph $\mathcal{G}$ is the minimum $\phi$ such that $\mathcal{G}$ has a TIM decomposition of width $\phi$.
\end{definition}

To parameterise by TIM width we apply a useful meta-algorithm; to do this we need to show our problem fits certain properties.
\begin{definition}[(k,X)-State]
A $(k,X)$-state on a vertex set $V$ is a pair $(l,c)$, where
$l : V \rightarrow X$ is a labelling of the vertices in $V$ using the labels from set $X$, and $c$ is a vector containing $k$ integers each of size at most a polynomial of $|V|$.
\end{definition}

A $(1,X)$-state for our problem has $X=Q$, the set of states in our automaton. We let all our vectors be one-dimensional zero vectors. Then each $(k,X)$-state can be thought of as an assignment $l$ of vertices to states in the automaton. We will use these states to show that the problem fits the properties required to apply the meta-algorithm.

\begin{definition}[(k,X,f)-Component-Exchangeable Temporally Uniform Problem] \label{def:kxf problem}
We say that a decision problem $P$ with input $x = (\mathcal{G},\beta)$ such that $\mathcal{G}$ has lifetime $\Lambda$ is \emph{$(k,X,f)$
component-exchangeable temporally uniform} if and only if there exist:
\begin{enumerate}
    \item a transition algorithm \textbf{Tr} that takes two labellings for the vertices of a connected, static graph $C$ with labels from the set $X$, the graph $C$, and the problem instance, runs in time
at most $f(|C|,x)$, and returns \textbf{true} or \textbf{false},
\item a starting algorithm \textbf{St}, a validity algorithm \textbf{Val}, and a finishing algorithm \textbf{Fin} that all
take a $(k,X)$-state of a connected static graph $C$ and the problem instance, run in time
at most $f(|C|,x)$, and return \textbf{true} or \textbf{false},
\item a vector $\textbf{v}_{upper}$ of $k$ integers, such that $x = (\mathcal{G},\beta)$ is a yes instance of $P$ if and only if there exists a sequence $s_0,...,s_{\Lambda}$ of
$(k,X)$-component states of each snapshot of $\mathcal{G}$ where
\begin{enumerate} 
    \item for each connected component $C_1$ of $G_1$, $\textbf{St}(s_0|_{C_1}, C_1,x) = \textbf{true}$;
\item for each connected component $C_\Lambda$ of $G_\Lambda$, $\textbf{Fin}(s_\Lambda |_{C_\Lambda},C_\Lambda,x) = \textbf{true}$;
\item $\textbf{Tr}(l_{t-1}|_{C_t},l_t|_{C_t},C_t,x) = \textbf{true}$ where $l_i$ is the labelling of vertices of state $s_i$, for all times $1 \leq t \leq \Lambda$ and connected components $C_t$ of $G_t$;
\item $\textbf{Val}(s_t|_{C_t},C_t,x) = \textbf{true}$ for all times $1 \leq t \leq \Lambda$ and connected components $C_t$ of $G_t$; and
\item the sum of vectors satisfies $\Sigma_{0\leq t\leq \Lambda} \Sigma_{C \in C_t} v_{st}(C) \leq \textbf{v}_{upper}$, where $v_{st}$ is the function $v$ in
the $(k,X)$-component state $s_t$ and $\leq$ denotes element-wise vector inequality.
\end{enumerate}
\end{enumerate}
\end{definition}

\begin{theorem}[\cite{enright2025families}]\label{thm:tim}
    Let $x=(\mathcal{G},\beta)$ be an instance of a $(k,X,f)$-component-exchangeable temporally uniform problem $P$ where $\mathcal{G}$ has $n$ vertices and lifetime $\Lambda$. Given a TIM decomposition of $\mathcal{G}$ with width $\phi$, we can determine whether $x$ is a yes-instance of $P$ in time $\mathcal{O}(n\Lambda|X|^{12\phi^3} (3b)^{12 k \phi^ 3} (3\Lambda n)^ {4k} \phi^9k^2f(\phi,x))$, where $\phi$ is the TIM width of $\mathcal{G}$ and $b$ is the maximum absolute value of any entry of a vector in a $(k,X)$-component state of $\mathcal{G}$.
\end{theorem}

We now apply this machinery to temporal role colouring. 
\begin{lemma} \label{lem:1Qf state}
   The \textsc{temporal role colouring} problem is a $(1,Q,f(C))$-component-exchangeable temporally uniform problem where $f(C)=\rho|C|^3$.
\end{lemma}
\begin{proof}
We want to show \textsc{temporal role colouring} is a $(1,Q,f(C))$-component-exchangeable temporally uniform problem where $f(C)=\rho|C|^3$. That is for $k=1$, $X=Q$ and $f(|C|,(\mathcal{G},A))=\rho|C|^3$ the conditions given in Definition~\ref{def:kxf problem} hold.
We first show that conditions $1$ and $2$ hold for our problem by construction.
First, for labellings $l_i,l_j: V(C) \rightarrow Q$ let \textbf{Tr}$(l_i,l_j,C,(\mathcal{G},A))=\textbf{true}$ if for all $v \in C(V)$ there exists a transition $\delta(l_i(v),c(N_C(v)))=l_j(v)$ and \textbf{false} otherwise. Checking the transition function takes time $\rho$ and finding neighbourhoods and their colours takes time $|C|^2$. Hence the time to check the transition function is $\rho|C|^3$ and Condition $1$ holds.

Next, we define \textbf{St}$(l, C,(\mathcal{G},A))$, as \textbf{true} if, for all $v \in V(C)$, we have $\delta(q_0,\varepsilon)=l_0(v)$, and \textbf{false} otherwise. This runs in time $\rho|C|$. We do not need an extra validity check so we let \textbf{Val}$(l,C,(\mathcal{G},A)) =\textbf{true}$ for all inputs. Finally, we define \textbf{Fin}$(l,C,(\mathcal{G},A))=\textbf{true}$ if for all $v\in C: l(v)\in F$, else \textbf{false}. This takes time $|C|$ to check. Therefore, as these all can be checked in time $\rho|C|^3$, Condition $2$ holds. 

Now we consider Condition $3$. We let $\textbf{v}_{upper}$ be the one dimensional zero vector. We want to show that $x=(\mathcal{G},A)$ is a yes instance if and only if there exists a sequence $s_0,...,s_\Lambda$ of $(1,X)$-component states of each snapshot of $\mathcal{G}$ with the properties stated.

First, we suppose there exists a sequence $s_0,..,s_\Lambda$ of $(1,X)$-component states of $\mathcal{G}$ such that:
\textbf{St}$(s_0|_{C_0},C_0,x)=$\textbf{true} for each connected component $C_0$ of $G_0$, \textbf{Fin}$(s_\Lambda|_{C_\Lambda},c_\Lambda,x)=\textbf{true}$ for each connected component $C_\Lambda$ of $G_\Lambda$ and, for all $1\leq t \leq \Lambda$, \textbf{Tr}$(l_{t-1}|_{C_t},l_t|_{C_t},C_t,x) = \textbf{true}$ for each connected component $C_t$ of $G_t$ .
Then for each vertex $v \in V$ by considering the sequence of connected components such that $v \in V(C_i)$ we have that $\delta(q_0,\varepsilon)=l_0(v)$, $\delta(s_\Lambda(v),N^\Lambda(v))\in F$ and, for all $1\leq t \leq \Lambda$, $\delta(l_{t-1}(v),c(N^t_{C_t}(v)))=l_t(v)$.
As each $C_t$ is a connected component we know that $N^t_{C_t}(v)=N^t(v)$. This gives us that the word  $c(N^1(v))c(N^2(v))...c(N^\Lambda(v))$ is accepted by the automaton. This holds for all $v \in V$ and is exactly the requirement for $(\mathcal{G},A)$ to be a yes instance.

Now we show the converse. Let $(\mathcal{G},A)$ be a yes instance. As \textbf{Val}  is always \textbf{true}, Condition $(d)$ holds vacuously. Similarly, Condition $(e)$ holds as all our vectors are zero. 
As it is a yes instance we have a colouring $r:V \rightarrow [k]$ such that for each vertex $v \in V$ there is an accepting path in $A$ given by $r(N^1(v))r(N^2(v))...r(N^\lambda(v))$  with states $q_1,...q_{\Lambda+1}$ such that $r(v,t)=c(q_t)$. We define the states $s_0,..s_\Lambda$ by $l_i(v)=q_{i+1}(v)$ and show the conditions hold. 
For Condition $(a)$ we know that the only transitions from $q_0$ are $\varepsilon$-transitions, and $q_0$ is our start state of the automaton, therefore $\delta(q_0,\varepsilon)=q_1(v)=l_0(v)$ must be a valid transition for all $v\in V$. That is for all $v\in V$ and therefore for all $V(C_1) \subseteq V$ we have that $\textbf{St}(s_0|_{C_1}, C_1,x) = \textbf{true}$.
For Condition $(b)$ we know that for a word to be accepted by an automaton it must end in an accepting state. In $A$ the accepting states are given by $F$, so we know $l_{\Lambda}(v)=q_{\Lambda+1}(v) \in F$ for all $v \in V$. That is for all $v \in V$ and therefore all $V(C_\Lambda) \subseteq V$ we have that \textbf{Fin}$(s_\Lambda|_{C_\Lambda},c_\Lambda,x)=\textbf{true}$. Finally for Condition $(c)$, we consider a vertex $v$ in connected components $C_1,..,C_\Lambda$. We know that for all $v\in V$ the path $r(N^1(v))r(N^2(v))...r(N^\Lambda(v))$ has states $q_1,...,q_{\Lambda+1}$, which means that each transition from state $q_t$ to state $q_{t+1}$ is labelled by the colours of the neighbourhood $N^t(v)$, that is, for all $1\leq t \leq \Lambda$, $\delta(q_t,r(N^t(v)))=q_{t+1}$. As $r(v,t)=c(q_t)$ and all neighbours of $v$ at time $t$ are in its connected component $C_t$ we can rewrite $r(N^t(v))$ as $c(N_{C_t}^t(v))$. Using this and the $l_i(v)=q_{i+1}(v)$ identity gives us that $\delta(l_{t-1}(v),c(N^t_{C_t}(v)))=l_t(v)$ for all times $1 \leq t \leq \Lambda$. As we showed this for an arbitrary $v\in V$ we have for all times $1\leq t \leq \Lambda$ and connected components $C_t$ that $\textbf{Tr}(l_{t-1}|_{C_t},l_t|_{C_t},C_t,x) = \textbf{true}$.
Therefore we have shown that \textsc{temporal role colouring} is a $(1,Q,\rho|C|^3)$-component-exchangeable temporally uniform problem.
\qed \end{proof}

Now we have shown everything we need to apply the meta-algorithm. 
\begin{theorem}
    For a temporal graph $\mathcal{G}$ with $n$ vertices and lifetime $\Lambda$ and a $k$-colour role automaton $A$ with $\rho$ states the \textsc{temporal role colouring} problem on $\mathcal{G}$ is solvable in time $\mathcal{O}((n\Lambda)^5(3\rho)^{12\phi^3}\rho\phi^{12})$ where $\phi$ is the tree-interval-membership width of $\mathcal{G}$.
\end{theorem}
\begin{proof}
    We prove this by direct application of Theorem~\ref{thm:tim}.
    From Lemma~\ref{lem:1Qf state} we know that \textsc{temporal role colouring} is a $(1,Q,f(C))$-component-exchangeable temporally uniform problem where $f(C)=\rho|C|^3$. Therefore we can apply Theorem~\ref{thm:tim}, with $k=1$, $|X|=|Q|=\rho$ and $f(\phi,x)=\rho\phi^3$. We set $b=1$ which is clearly an upper bound on the maximum absolute value of a vector in our algorithm. This
    gives us the result that we can determine whether $x$ is a yes instance in time $\mathcal{O}(n \Lambda \rho^{12\phi^3} (3)^{12\phi^3}(3\Lambda n)^4\phi^9\rho \phi^3)=\mathcal{O}((n\Lambda)^5(3\rho)^{12\phi^3}\rho\phi^{12})$.
\qed \end{proof}

\section{Treewidth}
In this section we give an fpt-algorithm for \textsc{temporal role colouring} parameterised by the number of colours of the automaton, the lifetime of the temporal graph and the treewidth of the underlying graph.

\begin{definition}
    [\cite{cygan2015parameterized}]\label{def:tree_decomp} 
    A \emph{tree decomposition} of a graph $G$ is a pair $(H,\mathcal{D})$ where $H$ is a tree and $\mathcal{D}=\{\mathcal{D}(b):b\in V(H)\}$ is a collection of subsets of $V(G)$, called \emph{bags}, such that:
\begin{enumerate}
    \item For every vertex $v \in V(G)$ there is at least one $b \in V(H)$ such that $v \in \mathcal{D}(b)$.
    \item For every edge $uv \in E(G)$ there exists a node $b \in V(H)$ such that $u,v \in \mathcal{D}(b)$.
    \item For every vertex $v \in V(G)$ the subgraph induced by the set of nodes $\{b: v \in \mathcal{D}(b)\}$ is connected.
\end{enumerate}
The \emph{width} $\omega$ of the tree decomposition is the size of the largest bag minus 1.
\end{definition}
 The \emph{treewidth} is the minimum width over all possible tree decompositions of a graph. We use \emph{nice} tree decompositions, which have restrictions on the relationships between adjacent bags. 
 \begin{definition}[\cite{cygan2015parameterized}]
We call a tree decomposition $(H,\mathcal{D})$ \emph{nice} if all leaves and the root of the tree contain empty bags and all non-leaf nodes are one of three types:
\begin{itemize}
    \item \emph{Introduce node}: a node $b$ with exactly one child $b'$ such that $\mathcal{D}(b)=\mathcal{D}(b')\cup \{v\}$, for some vertex $v \notin \mathcal{D}(b')$.
    \item \emph{Forget node}: a node $b$ with exactly one child $b'$ such that $\mathcal{D}(b)=\mathcal{D}(b')\setminus \{v\}$ for some vertex $v \in \mathcal{D}(b')$.
    \item \emph{Join node}: a node $b$ with exactly two children $b_1$ and $b_2$ such that $\mathcal{D}(b)=\mathcal{D}(b_1)=\mathcal{D}(b_2)$.
\end{itemize}
\end{definition}
 Any tree decomposition can, in polynomial time, be transformed into a nice tree decomposition without increasing the width, and without loss of generality we can assume any nice tree decomposition of width $\omega$ has $\mathcal{O}(\omega n)$ nodes~\cite{cygan2015parameterized}. We now state the main theorem of this section.

\begin{theorem} \label{thm:treewidth}
    For a temporal graph $\mathcal{G}$ and a $k$-colour role automaton $A$, \textsc{temporal role colouring} on instance ($\mathcal{G},A$) is solvable in time $\mathcal{O}(\Lambda n\omega(\rho^2+\omega k)2^{6k(\Lambda+1)(\omega+1)})$ where $\omega$ is the treewidth of the underlying graph $G_{\downarrow}$, $\Lambda$ is the lifetime of $\mathcal{G}$ and $\rho$ is the number of states of $A$.
\end{theorem}
We prove this by a dynamic program over a nice tree decomposition in which we compute all possible configurations for each node.
For any node $b\in V(H)$ the configuration of the node (often called state but we rename here for distinction from the states of the automaton) consists of:
\begin{itemize}
    \item A function $\alpha:\mathcal{D}(b) \times [\Lambda+1] \rightarrow [k]$ a colour for each time for each vertex in the bag including a final finishing time.
    \item A function $\beta: \mathcal{D}(b) \times [\Lambda] \rightarrow 2^{[k]}$ a set of colours for each time for each vertex in the bag; this tracks the colours of neighbours seen by each vertex at each time. 
\end{itemize}

We use $\mathcal{D}^+(b)$ to denote all vertices in the tree decomposition that are in the bag $\mathcal{D}(b)$ or bags at nodes below $b$. We consider a configuration valid if there is a colouring $\pi$ of all active and forgotten vertices $\pi: \mathcal{D}^+(b) \times [\Lambda+1] \rightarrow [k] $ such that:
\begin{enumerate} [label=\roman*.]
    \item $\pi|_{\mathcal{D}(b) \times [\Lambda+1]}=\alpha$
    \item for all $v \in \mathcal{D}(b)$ and $t \in [\Lambda]$, $\beta(v,t)=\pi(N^t(v)\cap\mathcal{D}^+(b))$
    \item for all $v \in \mathcal{D}^+(b) \setminus \mathcal{D}(b)$ we have that $ \pi(N^1(v))\pi(N^2(v))...\pi(N^\Lambda(v))$ is an accepting path through the automaton with states $q_1,q_2,...q_{\Lambda+1}$ such that $\pi(v,t)=c(q_t)$ for all $t \in [\Lambda+1]$.
\end{enumerate}
In this case we say the configuration is supported by $\pi$.

We prove Theorem \ref{thm:treewidth} by first considering the empty leaf nodes and then showing the validity conditions of each type of parent node.

\begin{lemmarep}\label{lem:leaf}
    A configuration $(\alpha,\beta)$ is valid at a leaf node if and only if:
    \begin{itemize}
        \item $\alpha$ is the empty function
        \item $\beta$ is the empty function.
    \end{itemize}
\end{lemmarep}
\begin{proof}
Let $(\alpha,\beta)$ be a valid configuration for the empty leaf node, then it is supported by a colouring $\pi:\mathcal{D}^+(b) \times [\Lambda+1] \rightarrow [k]$. As the bag at the leaf node is empty and there is nothing below the node we know that $\mathcal{D}^+(b)=\emptyset$ and $\pi$ is the empty function. Validity Conditions $i$ and $ii$ therefore give us that $\alpha$ and $\beta$ must both also be the empty function.

Now we let $\alpha$ and $\beta$ both be the empty function, we show that when $\pi$ is the empty function it supports $(\alpha,\beta)$ as a valid configuration. Clearly $\pi|_{\emptyset \times [\Lambda+1]}=\alpha$ and Condition $i$ holds. As $\mathcal{D}(b)=\mathcal{D}^+(b) \setminus \mathcal{D}(b)= \emptyset$, Conditions $ii$ and $iii$ hold vacuously. Therefore, $(\alpha,\beta)$ is a valid configuration. 
\qed \end{proof}

\begin{lemmarep} \label{lem:intro node}
    Let $b$ be an introduce node with child $b'$ such that $\mathcal{D}(b) \setminus \mathcal{D}(b')=\{u\}$. The $(\alpha,\beta)$ is a valid configuration for $b$ if and only if there exists a valid configuration $(\alpha',\beta')$ of $b'$ such that: 
    \begin{enumerate}
        \item $\alpha|_{\mathcal{D}(b')\times[\Lambda+1]}=\alpha'$
        \item $\beta(u,t)=\alpha(N^t(u) \cap \mathcal{D}(b))$
        \item for all $v\neq u$ and $t \in[\Lambda]$, $\beta(v,t)=
        \begin{cases}
            \beta'(v,t) \cup \{\alpha(u,t)\} &\text{if } t \in \lambda(uv)\\
            \beta'(v,t) &\text{otherwise}.
        \end{cases} $ 
    \end{enumerate}
\end{lemmarep}
\begin{proof}
Let $(\alpha,\beta)$ be a valid configuration for $b$, then there is a colouring $\pi$ of $\mathcal{D}^+(b)$ such that the validity conditions hold. 
We show that there exists a valid configuration $(\alpha',\beta')$ for $b'$ with the conditions required supported by $\pi'=\pi|_{\mathcal{D}^+(b') \times [\Lambda +1]}$.
Firstly we note that from validity of $(\alpha,\beta)$ we have that $\pi|_{\mathcal{D}(b) \times [\Lambda+1]}=\alpha$.
We let $\alpha'=\alpha|_{\mathcal{D}(b')\times[\Lambda+1]}$, which gives us that
 $\pi|_{\mathcal{D}(b')\times[\Lambda+1]}=\alpha'$ and as $\mathcal{D}(b') \subseteq \mathcal{D}^+(b)$ this gives us that $\pi'|_{\mathcal{D}(b')\times[\Lambda+1]}=\alpha'$.
So validity Condition $i$ holds when $\alpha|_{\mathcal{D}(b')\times[\Lambda+1]}=\alpha'$.

From validity Condition $ii$ of $b$ we know that, for all $v \in \mathcal{D}(b)$, $\beta(v,t)=\pi(N^t(v) \cap\mathcal{D}^+(b))$. 
We define $\beta'(v,t)=\pi'(N^t(v) \cap \mathcal{D}^+(b'))$. We note that $\pi'(N^t(v) \cap\mathcal{D}^+(b'))=\pi(N^t(v) \cap\mathcal{D}^+(b'))$ because $\pi'=\pi|_{\mathcal{D}^+(b') \times [\Lambda +1]}$.
First, we consider all vertex-time pairs $(v,t)$ such that $t \notin \lambda(uv)$. 
Then we have that, for all $t$, $N^t(v) \cap \mathcal{D}^+(b)=N^t(v)\cap \mathcal{D}^+(b')$ which gives us that $\beta'(v,t)=\beta(v,t)$ as required.
Now we consider vertex-time pairs $(v,t)$ such that $t \in \lambda(uv)$. We know from validity Condition $ii$ that $\beta(v,t)=\pi(N^t(v)\cap\mathcal{D}^+(b))$ and as $t \in \lambda(uv)$ we have that $u \in N^t(v) \cap \mathcal{D}^+(b)$. We can rewrite $\beta(v,t)$ as $\pi((N^t(v)\cap\mathcal{D}^+(b))\setminus\{u\}) \cup\{ \pi(u,t)\}= \pi((N^t(v)\cap\mathcal{D}^+(b))\setminus\{u\}) \cup \{\alpha(u,t)\}$.
We know that $\beta'(v,t)=\pi'(N^t(v) \cap \mathcal{D}^+(b'))=\pi(N^t(v) \cap \mathcal{D}^+(b'))$. As $\mathcal{D}(b) \setminus \mathcal{D}(b')=\{u\}$ we have that $\pi(N^t(v) \cap \mathcal{D}^+(b'))=\pi((N^t(v) \cap \mathcal{D}^+(b)) \setminus \{u\})$. 
Therefore, we have $\beta'(v,t)$ with Condition $ii$ holding such that $\beta(v,t)=\beta'(v,t)\cup \{\alpha(u,t)\}$ as required

Finally, we note that validity Condition $iii$ holds as $\mathcal{D}^+(b') \setminus \mathcal{D}(b') = \mathcal{D}^+(b) \setminus \mathcal{D}(b)$.
Hence all our validity conditions hold and $(\alpha',\beta')$ is a valid configuration for $b'$.

Let $(\alpha',\beta')$ be a valid configuration for $b'$ with the conditions stated. Let $\pi'$ be a colouring that supports the validity of the configuration. We will show there is a colouring $\pi^*$ that supports $(\alpha,\beta)$. We know from Condition $1$ that $\alpha|_{\mathcal{D}(b')\times[\Lambda+1]}=\alpha'$ and from validity Condition $i$ that $\pi'|_{\mathcal{D}(b') \times [\Lambda+1]}=\alpha'$ so we construct the colouring $\pi^*$ by $\pi^*(v,t)=\pi(v,t)$ for all $v \in \mathcal{D}^+(b')$ and $\pi^*(u,t)=\alpha(u,t)$.
We note this gives us that $\pi^*|_{\mathcal{D}^+(b') \times [\Lambda+1]}=\pi'$
and that $\pi^*|_{\mathcal{D}(b) \times [\Lambda+1]}=\alpha$ as required.

Next we consider the introduced vertex $u$. We note that $N^t(u) \cap \mathcal{D}(b)=N^t(u)\cap\mathcal{D}^+(b)$ due to the tree decomposition requirement that adjacent vertices have a common bag.
As, from Condition $2$, $\beta(u,t)= \alpha(N^t(u) \cap \mathcal{D}(b))$  we can see that $\beta(u,t)=\pi^*(N^t(u)\cap\mathcal{D}^+(b))$.
For all other vertices we have, from Condition $3$, that $\beta(v,t)=\beta'(v,t) \cup \{\alpha(u,t)\}$ if $t \in \lambda(uv)$ and $\beta'(v,t)$ otherwise. We know from validity Condition $ii$ that $\beta'(v,t)=\pi'(N^t(v)\cap\mathcal{D}^+(b'))$. If $t \notin \lambda(uv)$ then $N^t(v) \cap \mathcal{D}(b)=N^t(v)\cap\mathcal{D}(b')$ and $\pi^*(N^t\cap\mathcal{D}(b))=\pi'(N^t(v)\cap \mathcal{D}(b))$ so $\beta(v,t)=\pi^*(N^t\cap\mathcal{D}(b))$ as required.
If $t \in \lambda(uv)$ then $\beta(v,t)=\beta'(v,t) \cup \{\alpha(u,t)\}$, that is, $\beta(v,t)=\pi'(N^t(v)\cap\mathcal{D}^+(b')) \cup \pi^*|_{\mathcal{D}(b) \times [\Lambda+1]}(v,t)$.
From our construction of $\pi^*$ we can write this as $\beta(v,t)=\pi^*|_{\mathcal{D}^+(b') \times [\Lambda+1]}(N^t(v)\cap\mathcal{D}^+(b')) \cup \{\pi^*|_{\mathcal{D}(b) \times [\Lambda+1]}(u,t)$\} or $\beta(v,t)=\pi^*(N^t(v)\cap\mathcal{D}^+(b')) \cup \{\pi^*(u,t)\}=\beta$. 
Finally, we note that as $t \in \lambda(uv)$ and $u \in \mathcal{D}(b)$ we have that $N^t(v) \cap\mathcal{D}^+(b)=(N^t(v)\cap\mathcal{D}^+(b')) \cup \{u\}$ giving us that $\beta(v,t)=\pi^*(N^t(v)\cap\mathcal{D}^+(b))$ and validity Condition $ii$ holds.

Finally, we note that Condition $iii$ holds as $\mathcal{D}^+(b) \setminus \mathcal{D}(b) = \mathcal{D}^+(b') \setminus \mathcal{D}(b')$. Hence all validity conditions hold and $(\alpha,\beta)$ is a valid configuration for $b$.
\qed \end{proof}

\begin{lemmarep} \label{lem:forget node}
    Let $b$ be a forget node with child $b'$ such that $\mathcal{D}(b)=\mathcal{D}(b')\setminus \{u\}$ for some $u \in \mathcal{D}(b)$. Then $(\alpha,\beta)$ is a valid configuration for $b$ if and only if there exists a valid configuration $(\alpha',\beta')$ of $b'$ such that:
    \begin{enumerate}
        \item $\alpha=\alpha'|_{\mathcal{D}(b) \times [\Lambda+1]}$
        \item $\beta=\beta'|_{\mathcal{D}(b) \times [\Lambda]}$
        \item $\beta'(u,1)\beta'(u,2)...\beta'(u,\Lambda)$ is an accepting path through the automaton with states $q_1,q_2,...q_{\Lambda+1}$ such that $\alpha'(u,t)=c(q_t)$ for all $t \in [\Lambda+1]$.
    \end{enumerate}
\end{lemmarep}
\begin{proof}
    Let $(\alpha,\beta)$ be a valid configuration of $b$ and let $\pi$ be the colouring that supports it. We will show that there exists a valid configuration $(\alpha',\beta')$ for $b'$ with the conditions required also supported by $\pi$.
    As $\mathcal{D}(b') \subseteq \mathcal{D}^+(b)$ we can define $\alpha'$ as the restriction of $\pi$ to $\mathcal{D}(b') \times [\Lambda+1]$. As $\mathcal{D}(b) \subseteq \mathcal{D}(b')$ this gives us that $\alpha=\alpha'|_{\mathcal{D}(b) \times [\Lambda+1]}$. Hence $\alpha'$ is a valid function and Condition $1$ holds.

    We define $\beta'$ by $\beta'(v,t)=\pi(N^t(v)\cap\mathcal{D}^+(b'))$ for all $v \in \mathcal{D}(b)$.
    As $\mathcal{D}^+(b')=\mathcal{D}^+(b)$ this gives us that $\beta=\beta'|_{\mathcal{D}(b) \times [\Lambda]}$. Hence, $\beta'$ is such that both Condition $ii$ and Condition $2$ hold.

    Next, we consider the strictly forgotten vertex $u$. We know that $u \in \mathcal{D}^+(b) \setminus \mathcal{D}(b)$ so from validity Condition $iii$ we know that it has an accepting sequence given by $ \pi(N^1(u))\pi(N^2(u))...\pi(N^\Lambda(u))$ with states $q_1,q_2,...q_{\Lambda+1}$ such that $\pi(u,t)=c(q_t)$. From our definitions of $\beta'$ and $\alpha'$ this gives us that the word $\beta'(u,1)\beta'(u,2)...\beta'(u,\Lambda)$ is accepted by the automaton with states $q_1,...,q_{\Lambda+1}$ such that $\alpha'(u,t)=c(q_t)$ and Condition $3$ holds.
    
Finally, as $\mathcal{D}^+(b')\setminus\mathcal{D}(b') \subset \mathcal{D}^+(b)\setminus\mathcal{D}(b)$,  Condition $iii$ of $(\alpha',\beta')$ is inherited from Condition $iii$ of $(\alpha,\beta)$. Therefore, we have shown there exists a valid configuration $(\alpha',\beta')$ of $b'$ such that our conditions hold.

 Now we show the converse. Let $(\alpha',\beta')$ be a valid configuration for $b'$ supported by the colouring $\pi'$ such that our conditions hold. Then we show there is a valid configuration $(\alpha,\beta)$ for $b$ also supported by $\pi'$.
    From Condition $1$ we have that $\alpha=\alpha'|_{\mathcal{D}(b) \times [\Lambda+1]}$ and from validity of $\alpha'$ we know that $\pi'|_{\mathcal{D}(b') \times [\Lambda+1]}=\alpha'$, therefore we have that $\pi'|_{\mathcal{D}(b) \times [\Lambda+1]}=\alpha$ and Condition $i$ holds.
Next we see that from Condition $2$ we have that $\beta=\beta'|_{\mathcal{D}(b) \times [\Lambda]}$, which immediately gives us that Condition $ii$ holds.

Next, we consider the forgotten vertex $u$. From Condition $3$ we know that $ \beta'(u,1)\beta'(u,2)...\beta'(u,\Lambda)$ is an accepting path through the automaton with states $q_1,q_2,...q_{\Lambda+1}$ such that $\alpha'(u,t)=c(q_t)$, and as $u \notin \mathcal{D}(b)$ we know that all its neighbours are in $\mathcal{D}^+(b')$, therefore $\beta'(u,t)=\pi(N^t(u) \cup\mathcal{D}^+(b))=\pi(N^t(u))$. As $u \in \mathcal{D}(b')$ we have that $\alpha'(u,t)=\pi'(u,t)$. Hence Condition $iii$ holds for $u$. 
We know that $\mathcal{D}^+(b) \setminus \mathcal{D}(b)=\mathcal{D}^+(b') \setminus \mathcal{D}(b') \cup \{u\}$ and from validity of $(\alpha',\beta')$ Condition $iii$ holds for all vertices in $\mathcal{D}^+(b') \setminus \mathcal{D}(b')$, hence it holds for all vertices in $\mathcal{D}^+(b) \setminus \mathcal{D}(b)$.
We have therefore shown that $(\alpha,\beta)$ supported by $\pi'$ meets all three validity requirements, that is there is a valid configuration for $b$ when there is a valid configuration for $b'$ with the conditions as stated.
\qed \end{proof}

\begin{lemmarep} \label{lem:join node}
    Let $b$ be a join node with children $b_1$ and $b_2$ such that $\mathcal{D}(b)=\mathcal{D}(b_1)=\mathcal{D}(b_2)$. Then $(\alpha,\beta)$ is a valid configuration of $b$ if and only if there exist valid configurations $(\alpha_1,\beta_1)$ and $(\alpha_2,\beta_2)$ of $b_1$ and $b_2$ respectively such that:
    \begin{enumerate}
        \item $\alpha=\alpha_1=\alpha_2$
        \item $\beta(v,t)=\beta_1(v,t) \cup \beta_2(v,t)$ for all times $t \in [\Lambda]$ and vertices $v\in \mathcal{D}(b)$.
    \end{enumerate}
\end{lemmarep}
\begin{proof}
  Let $(\alpha,\beta)$ be a valid configuration of $b$ supported by the colouring $\pi$. We will show there exist valid configurations $(\alpha_1,\beta_1)$ and $(\alpha_2,\beta_2)$ for the states $b_1$ and $b_2$ with the conditions stated supported by $\pi_1=\pi|_{\mathcal{D}^+(b_1) \times [\Lambda+1]}$ and $\pi_2=\pi|_{\mathcal{D}^+(b_2) \times [\Lambda+1]}$ respectively.
First we define $\alpha_1$ and $\alpha_2$ by $\alpha_1(v,t)=\alpha_2(v,t)=\alpha(v,t)$ for all $v \in \mathcal{D}(b)$; we can define them this way as $\mathcal{D}(b)=\mathcal{D}(b_1)=\mathcal{D}(b_2)$. This immediately gives us that validity Condition $i$ holds for $\alpha_1$ and $\alpha_2$ supported by $\pi_1$ and $\pi_2$ as all three colouring functions coincide when restricted to $\mathcal{D}(b)\times[\Lambda+1]$.
We define $\beta_1(v,t)=\pi(N^t(v) \cap \mathcal{D}^+(b_1))$ and $\beta_2(v,t)=\pi(N^t(v) \cap \mathcal{D}^+(b_2))$; then, as $\mathcal{D}^+(b)=\mathcal{D}^+(b_1)\cup\mathcal{D}^+(b_2)$ for all $v \in \mathcal{D}(b)$, we have that $\beta(v,t)=\beta_1(v,t) \cup \beta_2(v,t)$.
Validity Condition $ii$ holds for $\beta_1$ and $\beta_2$ by definition.

As $\mathcal{D}(b_1)=\mathcal{D}(b)$ and $\mathcal{D}^+(b_1) \subseteq \mathcal{D}^+(b)$, the set $\mathcal{D}^+(b_1) \setminus \mathcal{D}(b_1)$ is a subset of $\mathcal{D}^+(b) \setminus \mathcal{D}(b)$.
Therefore, for all $v \in \mathcal{D}^+(b_1) \setminus \mathcal{D}(b_1)$ we have that  $ \pi(N^1(v))\pi(N^2(v))...\pi(N^\Lambda(v))$ is an accepting path through the automaton for with states $q_1,q_2,...q_{\Lambda+1}$ such that $\pi(v,t)=c(q_t)$ for all $t \in [\Lambda+1]$. By definition $\pi(v,t)=\pi_1(v,t)$ for all $v \in \mathcal{D}^+(b_1) \setminus \mathcal{D}(b_1)$. By the connectedness property of the tree decomposition we know that $N^t(v) \subseteq \mathcal{D}^+(b_1)$ and therefore $\pi(N^t(v))=\pi_1(N^t(v))$. Hence, Condition $iii$ holds for $(\alpha_1,\beta_1)$.
An equivalent argument shows that Condition $iii$ also holds for $(\alpha_2,\beta_2)$. 
We have therefore shown that there exist valid configurations for $b_1$ and $b_2$ such that our conditions hold.

Let $(\alpha_1,\beta_1)$ and $(\alpha_2,\beta_2)$ be valid configurations for $b_1$ and $b_2$ supported by colourings $\pi_1$ and $\pi_2$ respectively. We will consider an arbitrary parent configuration $(\alpha,\beta)$ such that the conditions of the lemma statement hold, and show that it is indeed a valid configuration for $b$. 
We build a colouring $\pi$ by setting $\pi(v,t)=\pi_1(v,t)$ for all $v \in \mathcal{D}^+(b_1)$ and $\pi(v,t)=\pi_2(v,t)$ for all $v \in \mathcal{D}^+(b_2) \setminus \mathcal{D}^+(b_1)$, we note that $\mathcal{D}^+(b_1) \cap \mathcal{D}^+(b_2)= \mathcal{D}(b_1)=\mathcal{D}(b_2)$ by the connectedness property of tree decompositions. Also by Condition $1$ we know that $\alpha_1=\alpha_2$ hence $\pi(v,t)=\pi_2(v,t)$ for all $v \in \mathcal{D}^+(b_2)$.
From validity Condition $i$ we have that $\pi_1|_{\mathcal{D}(b_1) \times [\Lambda+1]}=\alpha_1$ and from Condition $1$ we know that $\alpha=\alpha_1$. As $\mathcal{D}(b_1) \subseteq \mathcal{D}^+(b_1)$ we have that $\pi_1|_{\mathcal{D}(b_1) \times [\Lambda+1]}=\pi|_{\mathcal{D}(b_1) \times [\Lambda+1]}$, finally using the fact $\mathcal{D}(b_1)=\mathcal{D}(b)$ gives us that $\pi|_{\mathcal{D}(b) \times [\Lambda+1]}=\alpha$ and Condition $i$ holds.

From validity Condition $ii$ we have that $\beta_1(v,t)=\pi_1(N^t(v)\cap\mathcal{D}^+(b_1))$ and $\beta_2(v,t)=\pi_2(N^t(v)\cap\mathcal{D}^+(b_2))$. By our definition of $\pi$ we can rewrite these as $\beta_1(v,t)=\pi(N^t(v)\cap\mathcal{D}^+(b_1))$ and $\beta_2(v,t)=\pi(N^t(v)\cap\mathcal{D}^+(b_2))$.
From Condition $2$ we have that $\beta(v,t)=\beta_1(v,t) \cup \beta_2(v,t)$, which gives us that $\beta(v,t)=\pi(N^t(v)\cap\mathcal{D}^+(b_1)) \cup \pi(N^t(v)\cap\mathcal{D}^+(b_2))$ that is that $\beta(v,t)=\pi((N^t(v)\cap\mathcal{D}^+(b_1)) \cup (N^t(v)\cap\mathcal{D}^+(b_2)))=\pi(N^t(v)\cap(\mathcal{D}^+(b_1)\cup \mathcal{D}^+(b_2)))=\pi(N^t(v) \cap \mathcal{D}^+(b))$ and Condition $ii$ holds as required.

From validity Condition $iii$ we know that all vertices in $\mathcal{D}^+(b_1) \setminus \mathcal{D}(b_1)$ and $\mathcal{D}^+(b_2) \setminus \mathcal{D}(b_2)$ have valid accepting words. We know that $\mathcal{D}(b)=\mathcal{D}(b_1)=\mathcal{D}(b_2)$ and hence have that $\mathcal{D}^+(b) \setminus \mathcal{D}(b)=\mathcal{D}^+(b_1) \setminus \mathcal{D}(b_1) \cup \mathcal{D}^+(b_2) \setminus \mathcal{D}(b_2)$ and therefore Condition $iii$ holds for $(\alpha,\beta)$.

Therefore we have shown that when $(\alpha_1,\beta_1)$ and $(\alpha_2,\beta_2)$ are valid configurations with the conditions required so is $(\alpha,\beta)$.
    
\qed \end{proof}

Now we use these results to show there is an fpt-algorithm for the temporal role colouring problem.
\begin{proof}[Proof of Theorem~\ref{thm:treewidth}]
Given an instance $(\mathcal{G},A)$ of \textsc{temporal role colouring} we first compute a nice tree decomposition in time and then use a standard dynamic programming technique to determine if it is a yes instance by computing the set of valid configurations at each node from the leaves to the root. By our definition of a valid configuration we know that if there is a valid configuration of the empty root node then $\pi$ is a colouring of all vertices such that they have an accepting path through the automaton. Therefore, an instance $(A,\mathcal{G})$ is a yes instance if and only if there is a valid configuration at the root of the tree.

We now consider the running time for this dynamic program.
We assume throughout that we use a data structure that allows us to look up values $\alpha(v,t)$ and $\beta(v,t)$ for any $(v,t)$ in constant time and moreover that each set $\beta(v,t)$ is stored in lexicographic order.
We further assume that our temporal graph encoding (for example as a sequence of adjacency matrices) allows us to check existence of an edge at a given time in constant time.
First we note that the maximum number of valid configurations we would need to consider at any node is $\mathcal{O}(2^{2k(\Lambda+1)(\omega+1)})$, as there are $k^{(\Lambda+1)(\omega+1)}$ possible functions $\alpha$ and $2^{k\Lambda(\omega+1)}$ possible functions $\beta$. 
We consider how long it takes to check validity of each type of node.

\textbf{Leaf nodes:} by Lemma \ref{lem:leaf} validity can clearly be checked in constant time.

\textbf{Introduce nodes:} we check whether a given parent configuration $(\alpha,\beta)$ is valid by determining whether any configuration of its child satisfies the conditions of Lemma~\ref{lem:intro node}.
First, Condition $1$ involves checking equality of the function for all $v,t \in \mathcal{D}(b') \times [\Lambda+1]$, and each check requires constant time so overall this takes time $\mathcal{O}(\Lambda \omega)$.
For Condition $2$ we need to first compute $N^t(u) \cap \mathcal{D}(b)$.  To do this we check if each vertex in the bag is adjacent to $u$, taking time $\mathcal{O}(\omega)$.
Then we are checking equality of $\Lambda$ pairs of sets each of size at most $k$, making the total time to check Condition $2$ $\mathcal{O}(\Lambda k+ \omega)$.
Finally, to check Condition $3$ we consider each vertex-time pair $(v,t)$ and first check in constant time  whether $t \in \lambda(uv)$; if so we compute the set $\beta'(v,t) \cup \{\alpha(u,t)\}$ (maintaining lexicographic order) and check its equality with $\beta(v,t)$ in time $\mathcal{O}(k)$, and otherwise we check equality of $\beta'(v,t)$ and $\beta(v,t)$, again in time $\mathcal{O}(k)$. Overall checking Condition $3$ therefore takes time $\mathcal{O}(\Lambda \omega k)$.
Therefore, checking the conditions for any given child configuration takes time $\mathcal{O}(\Lambda\omega k)$ so determining validity of the parent configuration takes time $\mathcal{O}(\Lambda\omega k2^{2k(\Lambda+1)(\omega+1)})$.

\textbf{Forget nodes:} we check whether a given parent configuration $(\alpha,\beta)$ is valid by determining whether any configuration of its child satisfies the conditions of 
 Lemma~\ref{lem:forget node}. As for introduce nodes, checking Condition $1$ takes time $\mathcal{O}(\Lambda \omega)$.
For Condition $2$ we need to check equality of $\mathcal{O}(\Lambda \omega)$ sorted sets each of size at most $k$, taking time $\mathcal{O}(\Lambda\omega k)$. 
Finally, we consider checking Condition $3$. 
At time $0$ we construct in time $\mathcal{O}(\rho)$ the set of possible states $X_0$ given by $\delta(q_0,\varepsilon)$. At each other time step we construct the set of possible states $X_{t+1}$ as follows:
for each state $q_i$ in $X_t$ (of which there are at most $\rho$) we first check in constant time whether $c(q_i)=\alpha(v,t)$ and if it is we add any new states given by $\delta(q_i,\beta(v,t))$ to $X_{t+1}$, taking time $\mathcal{O}(\rho)$. Constructing $X_{t+1}$ from $X_t$ therefore takes time $\mathcal{O}(\rho^2)$.
This is done for each time, giving us the set $X_\Lambda$ in time $\mathcal{O}(\Lambda \rho^2)$. It then takes at most $\rho$ constant-time checks to see if any of the states in $X_\Lambda$ are in $F$, so the total time to check Condition $3$ is $\mathcal{O}(\Lambda \rho ^2)$. Therefore, checking the conditions for any given child configurations takes time $\mathcal{O}(\Lambda(\rho^2+\omega k))$ so determining validity of the parent configuration takes time $\mathcal{O}(\Lambda(\rho^2+\omega k)2^{2k(\Lambda+1)(\omega+1)})$.

\textbf{Join nodes:} we check whether a given parent configuration $(\alpha,\beta)$ is valid by determining whether any pair of configurations of its children satisfy the conditions of Lemma~\ref{lem:join node}.
Checking Condition $1$ takes time $\mathcal{O}(\Lambda \omega)$.
For Condition $2$ we have to calculate the union of the (sorted) sets $\beta_1(v,t)$ and $\beta_2(v,t)$ and compare this union to $\beta(v,t)$ for each time and vertex, which takes time $\mathcal{O}(\Lambda \omega k)$. Therefore, checking the conditions for any given pair of child configurations takes time $\mathcal{O}(\Lambda\omega k)$ so determining validity of the parent configuration takes time $\mathcal{O}(\Lambda\omega k2^{4k(\Lambda+1)(\omega+1)})$.


We can check validity of a given configuration of any parent node in time $\mathcal{O}(\Lambda(\rho^2+\omega k)2^{4k(\Lambda+1)(\omega+1)})$. Therefore, we can find all valid configurations at a node in time $\mathcal{O}(\Lambda(\rho^2+\omega k)2^{6k(\Lambda+1)(\omega+1)})$.
Recall that we may assume our nice tree decomposition has at most $\mathcal{O}(n\omega)$ nodes\cite{cygan2015parameterized}. 
We can therefore find the valid configurations at the root and therefore determine if $(\mathcal{G},A)$ is a yes instance in time $\mathcal{O}(\Lambda n\omega(\rho^2+\omega k)2^{6k(\Lambda+1)(\omega+1)})$.
 \qed \end{proof}

\section{Discussion and Future Work}
We have described a temporal analogue of role colouring, and provided both hardness and parameterised tractability results.  Because we have explored three sparse parameters, a natural next step is to consider dense but structured temporal graphs such as those with low temporal clique width, temporal modular-width and temporal neighbourhood diversity \cite{enright2024structural}.

Static role colouring is known to have uses varying from analysing food webs \cite{luczkovich2003defining} to improving web search results \cite{scripps2007node}.  We are therefore hopeful that temporal role colouring will be similarly useful in analysing time-varying systems. A first step might be to identify compact role colourings inspired by reality for data-derived temporal networks, such as those  hosted at SocioPatterns \cite{sociopatternsSocioPatternsx2013}.

Another possible direction for future work is the use of role colourability itself as a parameter.  That is, if a temporal graph admits a role colouring with a size $k$ automaton, can we solve otherwise-hard temporal graph problems using algorithms that are fixed parameter tractable by $k$?

Finally, there are graph classes on which static role colouring is tractable.
%
Are there any structurally-defined classes of temporal graphs or restrictions to the automata for which temporal role colouring can be solved efficiently?

\begin{credits}
\subsubsection{\ackname} Jessica Enright and Kitty Meeks are supported by EPSRC grants EP/T004878/1 and EP/V032305/1. 

\subsubsection{\discintname}
The authors have no competing interests to declare that are
relevant to the content of this article. 
\end{credits}

\pagebreak
\bibliographystyle{splncs04}
\bibliography{ref}

\end{document}